%% file: 00_paper.tex
\documentclass[12pt]{article}
\usepackage{amsmath, amsthm, amssymb}
\usepackage[left=1in, bottom=1in, right=1in, top=1in]{geometry}
\usepackage{xcolor}
\usepackage[comma,authoryear,round]{natbib}
\usepackage{soul}
\usepackage{graphicx}
\usepackage{dsfont}

\usepackage{float}
\usepackage{bbm}
\usepackage{authblk}
\usepackage{comment}
\usepackage{longtable}
\usepackage[font=normalsize]{caption}   
\usepackage[title]{appendix}
\usepackage{csquotes}
\usepackage{caption}
\usepackage{lipsum}

\newtheorem{proposition}{Proposition}[section]

\title{Modeling Multiple Irregularly Spaced Financial Time Series}

\author[a]{Chiranjit Dutta \thanks{Corresponding Author: chiranjit.dutta@uconn.edu}}
\author[b]{Nalini Ravishanker}
\author[c]{Sumanta Basu}

\affil[a,b]{Department of Statistics, University of Connecticut, Storrs, CT, USA}
\affil[c]{Department of Statistics and Data Science, Cornell University, Ithaca, NY, USA}

\date{\today}

\begin{document}
\maketitle

\begin{abstract}
In this paper we propose univariate volatility models for irregularly spaced financial time series by modifying the regularly spaced stochastic volatility models. We also extend this approach to propose multivariate stochastic volatility (MSV) models for multiple irregularly spaced time series by modifying the MSV model that was used with daily data. We use these proposed models for modeling intraday logarithmic returns from health sector stocks data obtained from Trade and Quotes (TAQ) database at Wharton Research Data Services (WRDS).
\end{abstract}



\input{01_Introduction}
\newpage
\bibliographystyle{abbrvnat} 
\bibliography{biblio-HFT-survey}

\end{document}

%% file: 01_Introduction.tex
\section{Introduction}
\label{1}

In finance, we often encounter time-series data that exhibits irregular spacing in time meaning that the time intervals between successive data points are not the same. This type of data is known as irregularly spaced time series data. We propose univariate and multivariate stochastic volatility models for irregularly spaced time series.  Specifically, we build on the gap time modeling idea of \cite{nagaraja2011autoregressive} to construct useful time series models that can help better understand volatility patterns in irregularly spaced financial time series. To achieve this, we modify existing stochastic volatility models that were originally designed for regularly spaced data. Additionally, we extend this approach to model multiple irregularly spaced time series by modifying multivariate stochastic volatility (MSV) model of \citep{chib2009multivariate} that was used with daily data.
High-frequency financial data are inherently irregularly spaced since trades can occur at any point in time. Moreover, the microstructure of financial markets, such as the methods of placing and executing orders, can also cause unevenly spaced data. For instance, in electronic markets, trades can be executed rapidly, and traders may use different strategies and algorithms that result in varying transaction frequencies.
The literature on dynamic statistical modeling is sparse. Practitioners interested in understanding dynamic evolution of stock properties require specialized models since the ones designed for regularly spaced data are inappropriate.
One of the early approaches to model the return series sampled at irregularly spaced time intervals set by the trade arrivals is ACD-GARCH model by \cite{ghysels1998garch}. This is a random coefficient GARCH, or doubly stochastic GARCH, where the stochastic durations between transactions determine the parameter dynamics for the GARCH equation. They proposed a Generalized Method of Moments (GMM) based two-step procedure for estimating the parameters. ACD-GARCH is quite cumbersome to estimate and also difficult to generalize to multiple time series. \cite{meddahi2006garch} proposed a GARCH type model for irregularly spaced data which is an exact discretization of continuous time stochastic volatility processes observed at irregularly spaced times. Their model combines the advantages of ACD-GARCH \cite{ghysels1998garch} and ACD \citep{Engle1998}. A continuous version of GARCH (COGARCH) is another way to model irregularly spaced time series data \citep{maller2008garch}. Recently, \cite{buccheri2021score} propose to model intraday log-prices through a multivariate local-level model with score-driven covariance matrices and to treat asynchronicity as a missing value problem.

\textcolor{red}{More details are needed}

\textcolor{red}{Write a short description about each section?}


\section{Volatility Models for Univariate Irregularly Spaced Time Series}


Let $\{r_{t_j}\}$ be a sequence of log-returns of a financial asset, where 
$t_{j}$ denotes the  time  of the $j^{th}$transaction. Let   
and $g_{j} = t_{j} - t_{j-1}, j>1$ be the known gap times between consecutive returns. Unlike regularly spaced time series (such as daily returns), for irregularly spaced time series (such as transaction level intra-day returns), the gap times $g_{j},~j>1$ are not the same.

\subsection{Irregularly spaced stochastic volatility (IR-SV) Model} \label{sec:IR-SV_properties}

We define an irregular stochastic volatility (IR-SV) model for $r_{t_j}$ as 
\begin{align}
        r_{t_{j}} & = \exp\Bigg(\frac{h_{t_{j}}}{2}\Bigg) \epsilon_{t_{j}}, \quad \epsilon_{t_{j}} \sim N(0,1), \label{Gap_SV1}\\
        h_{t_{1}} & = \mu + \eta_{t_{1}} , \quad \eta_{t_{1}} \sim N\Big(0,\dfrac{\sigma_{\eta}^2}{1-\phi^2}\Big), \label{Gap_SV2} \\
        h_{t_{j}} & = \mu + \phi^{g_j} (h_{t_{j-1}} - \mu) + \eta_{t_{j}} , \quad \eta_{t_{j}} \sim N\Bigg(0,\dfrac{\sigma_{\eta}^2 (1-\phi^{2g_j})}{1-\phi^2}\Bigg), \text{ for $j>1$}, \label{Gap_SV3}
\end{align}
where $\phi$ is the persistence parameter and $\mu$ is the location parameter.
In \eqref{Gap_SV2}-\eqref{Gap_SV3}, we assume that the log-volatility process $\{h_{t_j}\}$ is 
a stationary Gaussian autoregressive (AR) process with $|\phi| < 1$, and $\phi \neq 0$. We also assume that the gap times $g_j,~j>1$ are  bounded away from $0$. 

The formulation of the IR-SV model in \eqref{Gap_SV1}-\eqref{Gap_SV3} can be viewed as an extension to volatility modeling of the stationary gap time autoregressive (gap-AR) model 
for an irregularly spaced time series 
described in \cite{nagaraja2011autoregressive}. 



\begin{proposition} $h_{t_{j}}$ is a weakly stationary process and the unconditional distribution of $h_{t_{j}}$ is
\begin{equation*}
    h_{t_j} \sim \text{N}\Big(\mu,\frac{\sigma_{\eta}^2}{1-\phi^{2}}\Big),
\end{equation*}
\end{proposition}

\begin{proof}
Let $x_{t_{j}} = h_{t_{j}} - \mu , \forall j \in \mathbb{Z}^{+}$. From \eqref{Gap_SV1}-\eqref{Gap_SV3}, $x_{t_{j}}$ can be recursively written as  
\begin{equation}
    x_{t_{j}} = \sum_{k = 1}^{j-1} \phi^{\sum_{i= k+1}^{j} g_{i}} \eta_{t_{k}} + \eta_{t_{j}}
\end{equation}
Hence it follows that $E(x_{t_{j}}) = 0, \forall j \in \mathbb{Z}^{+}$. For $l=0,1,\cdots$, the gap time $l= t_{j+l} - t_{j} = \sum_{i=j+1}^{j+l} g_{i}$ and the covariance of $x_{t_{j}}$ and $x_{t_{j+l}}$ is
\begin{align*}
    Cov(x_{t_{j}},x_{t_{j+l}})& = \text{E} \left[\Bigg(\sum_{k = 1}^{j-1} \phi^{\sum_{i= k+1}^{j} g_{i}} \eta_{t_{k}} + \eta_{t_{j}}\Bigg) \Bigg(\sum_{k = 1}^{j+l-1} \phi^{\sum_{i= k+1}^{j+l} g_{i}} \eta_{t_{k}} + \eta_{t_{j+l}} \Bigg)\right] \\
    &= \sum_{k=1}^{j-1} \phi^{\sum_{i = k+1}^{j} g_{i} + \sum_{i = k+1}^{j+l} g_{i}} \text{E} \left[\eta^2_{t_{k}}\right] + \text{E} \left[\eta^2_{t_{j}}\right]\\
    &= \dfrac{\sigma_{\eta}^2}{1-\phi^2}  \left[\phi^{2\sum_{i=2}^{j} g_{i} + l} + \sum_{k=2}^{j-1} \phi^{2\sum_{i=k+1}^{j} g_{i}+l} (1-\phi^{2g_{k}}) + \phi^{l}(1-\phi^{2g_{j}})\right]\\
    &= \dfrac{\sigma_{\eta}^2 \phi^{l}}{1-\phi^2}  \left[\phi^{2\sum_{i=2}^{j}g_{i}} - \phi^{2\sum_{i=2}^{j}g_{i}} + \phi^{2\sum_{i=3}^{j}g_{i}} + \cdots + \phi^{2g_{j}} - \phi^{2g_{j}} + 1 \right] \\
    &= \dfrac{\sigma_{\eta}^2 \phi^{l}}{1-\phi^2} 
\end{align*}
Thus $Cov(x_{t_{j}},x_{t_{j+l}})$ is a function of gap time $l$ only and the series $x_{t_{j}}$ and hence $h_{t_{j}}$ is a weakly stationary process.
By stationarity of $h_{t_{j}}$ and (\ref{Gap_SV3}),
\begin{align*}
    &\text{Var}(h_{t_{j}}) = \phi^{2g_{j}}\text{Var}(h_{t_{j}}) + \text{Var}(\eta_{t_{j}})\\
    &\implies \text{Var}(h_{t_{j}}) = \dfrac{\sigma_{\eta}^2}{1-\phi^2}
\end{align*}
Also, $\text{E}(h_{t_{}j}) = \mu$ and hence the unconditional distribution of $h_{t_{j}}$ is
$\text{N}\Big(\mu,\frac{\sigma_{\eta}^2}{1-\phi^{2}}\Big)$.
\end{proof}

\begin{proposition}
We discuss properties of the distribution of $\{r_{t_j}\}$.    
\begin{enumerate}
    \item  The expectation and variance of the squared returns is
    \begin{align} \label{Eq:Var_IR_SV}
    &\text{E}\big(r_{t_j}^2 \big) =   \exp\Big( \mu + \frac{\sigma_{\eta}^2}{2(1-\phi^2)}\Big)\\ 
    &\text{Var}\big(r_{t_j}^2 \big) =   \exp\Big(2\mu +\dfrac{\sigma^2_{\eta}}{1-\phi^2}\Big)\Bigg(3\exp\Big(\dfrac{\sigma^2_{\eta}}{1-\phi^2}\Big) - 1\Bigg) 
\end{align}
    \item Let the gap time $l= t_{j+l} - t_{j} = \sum_{i=j+1}^{j+l} g_{i}$, $l=0,1,\cdots$, the autocovariance function of $r_{t_{j}}^2$ and $r_{t_{j+l}}^2$ is
    \begin{align} \label{Eq:Cov_IR_SV}
    \text{Cov}\big(r_{t_j}^2,r_{t_{j+l}}^2 \big) =   \exp\Big( 2\mu + \frac{\sigma_{\eta}^2}{1-\phi^2}\Big) \left[\exp\Big(\dfrac{\sigma_{\eta}^2 \phi^{l}}{1-\phi^2}\Big) - 1\right]\\ 
    \end{align}

\item The kurtosis of the returns is
\begin{equation} \label{Eq:kurtosis_IR_SV}
    \text{K}\big(r_{t_j}\big) =  3\exp\Big(\dfrac{\sigma^2_{\eta}}{1-\phi^2}\Big)
\end{equation}

\item The sequence of squared returns $r^{2}_{t_j}$  is a stationary process.

\end{enumerate}
\end{proposition}

\begin{proof}
The first two even moments of $r_{t_j}$ are
\begin{equation*}
\begin{aligned}
    \text{E}\big(r_{t_j}^2 \big) = \text{E}\big(\exp(h_{t_j}) \epsilon_{t_j}^2\big) =\text{E}\big(\exp(h_{t_j})\big) \text{E}\big(\epsilon_{t_j}^2\big) = \exp\Big( \mu + \frac{\sigma_{\eta}^2}{2(1-\phi^2)}\Big),\\
    \text{E}\big(r_{t_j}^4 \big) = \text{E}\big(\exp(2h_{t_j}) \epsilon_{t_j}^4\big) =\text{E}\big(\exp(2h_{t_j}\big) \text{E}\big(\epsilon_{t_j}^4\big) = 3\exp\Big( 2\mu + \frac{ 2\sigma_{\eta}^2}{1-\phi^{2}}\Big).
\end{aligned}
\end{equation*}
Then, 
\begin{equation*}
\begin{aligned}
    \text{Var}\big(r_{t_j}^2 \big) = & 3\exp\Big( \mu + \frac{ 3\sigma_{\eta}^2}{2(1-\phi^{2})}\Big) - \exp\Big( 2\mu + \frac{\sigma_{\eta}^2}{(1-\phi^2)}\Big)\\
    =&\exp\Big(2\mu + \dfrac{\sigma_{\eta}^2}{1-\phi^2}\Big)\Big(3\exp\Big(\dfrac{\sigma_{\eta}^2}{1-\phi^2}\Big) - 1\Big), 
\end{aligned}
\end{equation*}
and the kurtosis is 
\begin{equation*}
    \text{K}\big(r_{t_j}\big) = \dfrac{\text{E}\big(r_{t_j}^4 \big)}{\text{E}\big(r_{t_j}^2 \big)^2} = 3 \exp\Big(\frac{\sigma_{\eta}^2}{1-\phi^{2}}\Big).
\end{equation*}
The kurtosis is always greater than $3$ as long as $\sigma_{\eta}^2 > 0$.
For $l=0,1,\cdots$, the autocovariance function of $r_{t_{j}}^2$ is
\begin{equation*}
    \begin{aligned}
        \text{Cov}(r^{2}_{t_{j+l}},r^{2}_{t_{j}}) 
        &= \text{E}(r^{2}_{t_{j+l}}r^{2}_{t_{j}}) - \text{E}(r^{2}_{t_{j+l}})\text{E}(r^{2}_{t_{j}})  \\
        & = \text{E} (\exp(h_{t_{j}} + h_{t_{j+l}}) \epsilon_{t_{j}}^2 \epsilon_{t_{j+l}}^2) - \exp\Big( 2\mu + \dfrac{\sigma_{\eta}^2}{1-\phi^2}\Big)
    \end{aligned}
\end{equation*}
By independence of $\epsilon_{t_{j}}$'s and $h_{t_{j}}$'s, we have 
\begin{equation}\label{Eq:Autocovariance_IR_SV}
    \begin{aligned}
        \text{Cov}(r^{2}_{t_{j+l}},r^{2}_{t_{j}}) &=\text{E} (\exp(h_{t_{j}} + h_{t_{j+l}})) - \exp\Big( 2\mu + \dfrac{\sigma_{\eta}^2}{1-\phi^2}\Big).
    \end{aligned}
\end{equation}
We note that 
\begin{align*}
    \text{Var}(h_{t_{j}} + h_{t_{j+l}}) &= \text{Var}(h_{t_{j}}) + \text{Var}(h_{t_{j+l}}) + 2\text{Cov}(h_{t_{j}}, h_{t_{j+l}})\\
    & = \dfrac{2\sigma_{\eta}^2 (1+\phi^{l})}{1-\phi^2}
\end{align*}
By normality of $h_{t_{j}}$, it follows that $h_{t_{j}} + h_{t_{j+l}} \sim N\Big(2\mu, \dfrac{2\sigma_{\eta}^2 (1+\phi^{l})}{1-\phi^2}\Big)$.
Therefore (\ref{Eq:Autocovariance_IR_SV}) becomes
\begin{align}\label{Eq:Autocovariance_IR_SV1}
        \text{Cov}(r^{2}_{t_{j+l}},r^{2}_{t_{j}}) 
        &= \exp\Big( 2\mu + \frac{\sigma_{\eta}^2}{1-\phi^2}\Big) \left[\exp\Big(\dfrac{\sigma_{\eta}^2 \phi^{l}}{1-\phi^2}\Big) - 1\right].
\end{align}
Therefore, $r_{t_{j}}^2$ is a stationary process.
\end{proof}


\subsection{Bayesian Estimation of the IR-SV Model} \label{sec:IR-SV_Bayesian}

Let $\boldsymbol{r} = (r_{t_{1}},\cdots,r_{t_{T}})$ be the set of irregularly spaced observed returns following the IR-SV model in 
\eqref{Gap_SV1}-\eqref{Gap_SV3}. Let $\boldsymbol{\theta} = (\sigma_{\eta}^2,\phi,\mu)$ be the set of model parameters. 
For $j=1,\cdots,T$, we can rewrite 
\eqref{Gap_SV1}-\eqref{Gap_SV3} as follows: 

\begin{align}
         r_{t_{j}}|h_{t_{j}},\mathbf{\theta}  & \sim N\big(0, \exp(h_{t_{j}})\big) \label{Eq:Conditional dist SV1}\\
         h_{t_{1}}  & \sim N\big(\mu,\dfrac{\sigma^2_{\eta}}{1-\phi^2}\big)  \label{Eq:Conditional dist SV2} \\   
         h_{t_{j}}|h_{t_{j-1}}  & \sim N\Big(\mu + \phi^{g_{j}}  (h_{t_{j-1}} - \mu) ,\dfrac{\sigma^2_{\eta} (1-\phi^{2g_{j}})}{1-\phi^2}\Big), \quad j>1. \label{Eq:Conditional dist SV3}    
\end{align}
We assume the following priors:
\begin{align*}
        &\dfrac{\phi+1}{2} \sim B(20,1.5),\\
        &\dfrac{1}{\sigma_{\eta}^2} \sim G(2.5,0.025),\\ 
        &\mu \sim N(0,10), \label{Eq:IR-SV_priors}
\end{align*}
where $B$ denotes the beta distribution, $G$ denotes the gamma distribution and $N$ is the normal distribution. 
Let $p(\boldsymbol{\theta}) = p(\phi) p(\sigma^2_{\eta}) p(\mu)$ denote the joint prior of all the parameters, assuming independence. The joint posterior distribution of the hidden volatility states $\boldsymbol{h} =(h_{t_{1}},\cdots,h_{t_{T}})$ and the parameters $\boldsymbol{\theta} = (\sigma_{\eta}^2,\phi,\mu)$ is obtained by Bayes' rule as
\begin{equation}
    p(\boldsymbol{h},\boldsymbol{\theta}| \boldsymbol{r}) \propto p(\boldsymbol{\theta}) \prod_{j=1}^{T} p(r_{t_{j}}|h_{t_{j}},\boldsymbol{\theta}) p(h_{t_{j}}|h_{t_{j-1}},\boldsymbol{\theta}),
\end{equation}
where 
$p(r_{t_{j}}|h_{t_{j}},\boldsymbol{\theta})$ and  $p(h_{t_{j}}|h_{t_{j-1}},\boldsymbol{\theta})$ follows from \eqref{Eq:Conditional dist SV1}-\eqref{Eq:Conditional dist SV3}.
We employ a Metropolis-Hastings adaptive random-walk sampler for each of the parameters with a univariate normal proposal distribution within \textit{R} package \texttt{NIMBLE} package \citep{de2017programming} to generate samples from their respective posterior distribution.

\subsection{Simulation Study}\label{IR-SV_simulation}
We demonstrate the accuracy of the Bayesian estimation of the IR-SV model parameters under a few different simulation setups. We generated $100$ sets (replicates) of zero mean log returns data, each of length $T=5000$.
The gap times $g_{j}, 1 \leq j \leq T$ are generated as follows:
\begin{itemize}
    \item Generate $g_{j} \sim \mathcal{P}(\lambda = 3)$ and $g_{j} >0, \forall \ j$ where $\mathcal{P}(\lambda)$ is the Poisson distribution with mean $\lambda$. 
    \item Scale the $g_{j}$'s such that $0< g_{j} \leq 1 ,\forall j$
\end{itemize}
In our empirical analysis we have observed the mean gap times to be around 3 seconds with $10\%$ of them greater than 5 seconds and since the persistence parameter $|\phi|<1$, we scale the gap times such that $g_{j} \in (0,1]$ to ensure $\phi^{g_{j}}$ is bounded away from 0, $\forall j$. To study the effect of persistence parameter $\phi$ on the estimation accuracy, in all three scenarios we generated zero mean log returns from the IR-SV model in 
(\ref{Gap_SV1}-\ref{Gap_SV3}) with the true parameter values of $\mu = -9$, $\sigma =0.8$. In scenario 1, the true value of the persistence parameter $\phi$ is 0.2, in scenario 2, true value of $\phi$ is 0.6 and in scenario 3, true value of $\phi$ is 0.9.  
\begin{table}[H]
\begin{center}  
\begin{tabular}{ccccc}
\hline \hline
Parameter & True Value & Mean & Q(2.5\%) & Q(97.5\%) \\
\hline \hline
$\mu$ & -9.0000 & -8.9997 & -9.1117 & -8.8627 \\
$\phi$ & 0.2000 & 0.2438 & 0.1565 & 0.3346 \\
$\sigma$ & 0.8000 & 0.7918 & 0.7366 & 0.8551 \\
\hline \hline
\end{tabular}
\end{center}
\caption{True values and posterior estimates of parameters for Scenario 1 from the IR-SV model.}
\label{Table:Simulation_Scenario_1_IR_SV}
\end{table}
\begin{table}[H]
\begin{center}  
\begin{tabular}{ccccc}
\hline \hline
Parameter & True Value & Mean & Q(2.5\%) & Q(97.5\%) \\
\hline \hline
$\mu$ & -9.0000 & -8.9833 & -9.2253 & -8.7326 \\
$\phi$ & 0.6000 & 0.6253 & 0.5515 & 0.7056 \\
$\sigma$ & 0.8000 & 0.7832 & 0.7026 & 0.8769 \\
\hline \hline
\end{tabular}
\end{center}
\caption{True values and posterior estimates of parameters for Scenario 2 from the IR-SV model.}
\label{Table:Simulation_Scenario_2_IR_SV}
\end{table}
\begin{table}[H]
\begin{center}  
\begin{tabular}{ccccc}
\hline \hline
Parameter & True Value & Mean & Q(2.5\%) & Q(97.5\%) \\
\hline \hline
$\mu$ & -9.0000 & -8.9872 & -9.9697 & -8.0021 \\
$\phi$ & 0.9000 & 0.8960 & 0.8365 & 0.9458 \\
$\sigma$ & 0.8000 & 0.7860 & 0.7059 & 0.8818 \\ 
\hline \hline
\end{tabular}
\end{center}
\caption{True values and posterior estimates of parameters for Scenario 3 from the IR-SV model.}
\label{Table:Simulation_Scenario_3_IR_SV}
\end{table}
We run 520,000 MCMC iterations and discard the first 20,000 as burn-in and thinned every $1000^{th}$ sample to reduce autocorrelation between MCMC samples. Convergence for the parameters are assessed using trace and posterior density plots, \textbf{autocorrelation and thinning} (not shown here).
In Table \ref{Table:Simulation_Scenario_1_IR_SV}, Table \ref{Table:Simulation_Scenario_2_IR_SV}, and Table \ref{Table:Simulation_Scenario_3_IR_SV}, we report the posterior sample means of the parameters along with their true values and $95\%$ credible intervals, averaged the 100 data sets (replicates). The  true values of the parameters lie inside the $95\%$ credible intervals for all parameters in all three scenarios.

\section{Multivariate Stochastic Volatility Models for Irregularly Spaced Financial Time Series (IR-MSV)}


We propose a multivariate stochastic volatility model to fit irregularly spaced synchronized intraday zero mean log returns for multiple assets by modifying the basic multivariate stochastic volatility (BMSV) model \citep{chib2009multivariate} and using the gap time idea of \cite{nagaraja2011autoregressive} for handling the latent state.

\subsection{Refresh Time Sampling}\textcolor{red}{Why synchronization is necessary?}
In this section we describe refresh time sampling which allows synchronization of high frequency prices from multiple stocks. In high frequency trading (HFT) conducting multivariate analysis is difficult since assets do not trade on a fixed grid, trades and quotes don't arrive synchronously. Hence synchronization of HFT data from multiple assets is necessary. As in Section (2.5.3), we briefly describe the refresh sampling procedure for synchronization of high frequency prices from multiple stocks.

Assume there are $p$ stocks, and trading time of the $i^{th}$ stock is given by $t_{i\ell}$, $\ell=1,\cdots,n_i$ , $i=1,\cdots,p$. For a given time $t$, define $N_{t}^{i}=$ the number of $t_{i \ell} \leq t$, $\ell=1, \cdots, n_{i}$, which counts the number of distinct data points $t_{i \ell}$ available for the $i^{th}$ asset up to time $t$. 
The first refresh time is defined as $\tau_{1}=\max \left\{t_{11}, \ldots, t_{p 1}\right\}$, which is the first time taken to trade all assets and refresh their posted prices. The subsequent refresh times are defined as follows. Given the $j^{th}$ refresh time $\tau_{j}$, define the $(j+1)^{th}$ refresh time
\begin{equation}
 \tau_{j+1}=\max \left\{t_{1, N_{\tau_{j}}^{1}+1}, \ldots, t_{p, N_{\tau_{j}}^{p}+1}^{p}\right\}
\end{equation}
Suppose there are $m$ refresh time points $\tau_{1}, \cdots, \tau_{m}$. Intuitively, $\tau_{2}$ is the second time when all the assets are traded and their prices are refreshed.  
In our empirical analysis we consider three health sector stocks BMY, CVS and MDT traded on NYSE on $24^{th}$ June, $2016$. We consider the high frequency prices for these three stocks and illustrate the refresh time sampling procedure for synchronization of multivariate analysis which we discuss later. Figure \ref{fig:refresh_time_sampling} illustrates the refresh time sampling idea and in this example we have $\tau_1 = 09:45:01.00$,$\tau_2 = 09:45:02.00$, $\tau_3 = 09:45:03.00$, $\tau_4 = 09:45:06.00$ and $\tau_5 = 09:45:07.00$ are the first five refresh times. 
\begin{figure}[H]
    \centering
    \includegraphics[scale=0.70]{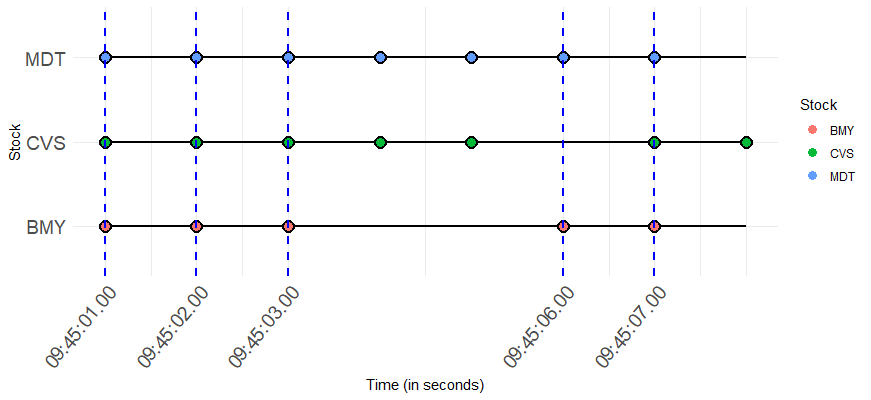}
    \caption{ Refresh Time Sampling for BMY, CVS and MDT. The vertical dotted blue lines represent the refresh time points on $24^{th}$ June, 2016}
    \label{fig:refresh_time_sampling}
\end{figure}

\subsection{IR-MSV Model Formulation}

Let $\boldsymbol{r}_{t_{j}} = (r_{1t_{j}},r_{2t_{j}},\cdots,r_{pt_{j}})^{\prime}$, $j = 1,\cdots,T$ be a $p$-dimensional vector of zero mean log returns of $p$ financial assets observed at irregularly spaced time points $t_{1},\cdots t_{T}$. Here, $t_{j}$ denotes the time 
of the $j^{th}$ observation, and $g_{j} = t_{j} - t_{j-1}, j>1$ denotes the known gap time between consecutive observations. We introduce irregular multivariate stochastic volatility (IR-MSV) model as
\begin{equation}\label{Eq:Irregular_MSV}
        \boldsymbol{r}_{t_{j}} = \mathbf{H}_{t_{j}}^{1/2} \boldsymbol{\epsilon_{t_{j}}}, \quad \boldsymbol{\epsilon_{t_{j}}} \sim \text{MVN}(\boldsymbol{0},\mathbf{R}),
\end{equation}
where,
\begin{equation}\label{Eq:Irregular_MSV1}
    \begin{aligned}
    &\mathbf{H}_{t_{j}} = \text{diag}\big(\exp(h_{1,t_{j}}),\cdots, \exp(h_{p,t_{j}})\big), \quad j=1,\cdots,T, \\
     & h_{i,t_{1}} = \mu _{i} + \eta_{i,t_{1}},\quad \eta_{i,t_{1}} \sim \text{N}\bigg(0,\dfrac{\sigma_{i}^2}{1-\phi_{i}^2}\bigg), \quad i=1,\cdots,p,\\
    & h_{i,t_{j}} = \mu _{i} + \phi_{i}^{g_j} (h_{i,t_{j-1}} - \mu _{i}) + \eta_{i,t_{j}}, \quad \eta_{i,t_{j}} \sim \text{N}\bigg(0,\dfrac{\sigma_{i}^2 (1-\phi_{i}^{2g_{j}})}{1-\phi_{i}^2}\bigg), \quad i=1,\cdots,p; j=1,\cdots,T,
    \end{aligned}
\end{equation}
where $\phi_{i}$ is the persistence parameter constrained by $|\phi_{i}| < 1$, $i=1,\cdots,p$, $\mu_{i}$ is the location parameter for each asset $i=1,\cdots,p$ and $\mathbf{R}$ is the correlation matrix of the observation errors. Following \cite{nagaraja2011autoregressive} and our discussion in the previous section, we observe that $\{h_{i,t_{j}}\}$ is a stationary process for each $i = 1,\cdots,p$. Since the error correlation matrix $\mathbf{R}$ is assumed to remain constant over time, we refer to (\ref{Eq:Irregular_MSV}) as a constant correlation IR-MSV model. 

\begin{proposition}
Let $\boldsymbol{r}^2_{t_{j}} = (r^2_{1t_{j}},r^2_{2t_{j}},\cdots,r^2_{pt_{j}})^{\prime}$, $j = 1,\cdots,T$ be the vector of squared returns. We discuss properties of the distribution of $\{\boldsymbol{r}_{t_j}^2\}$.     
\begin{enumerate}
    \item Let $\text{E}(\boldsymbol{r}^2_{t_{j}}) = \boldsymbol{m} = (m_{1},\cdots,m_{p})'$ be the $p\times 1$ vector of expected values where $m_{i}$ is
    \begin{align} \label{Eq:Var_IR_MSV}
    &m_{i} = \text{E}\big(r_{i,t_j}^2 \big) =   \exp\Big( \mu_{i} + \frac{\sigma_{i}^2}{2(1-\phi_{i}^2)}\Big), \quad i=1,\cdots,p\\ 
\end{align}
    \item Let $\text{Cov}(\boldsymbol{r}^2_{t_{j}}) = \boldsymbol{\Sigma}$ be the $p \times p$ variance-covariance matrix of squared returns vector. The $(i,k)^{th}$ element of $\boldsymbol{\Sigma}$ is
\begin{equation}
\text{Cov}(r_{i,t_{j}}^2,r_{k,t_{j}}^2)=\boldsymbol{\Sigma}_{ik}=
    \begin{cases}
        \exp\Big(2\mu_{i} +\dfrac{\sigma^2_{i}}{1-\phi_{i}^2}\Big)\Bigg(3\exp\Big(\dfrac{\sigma^2_{i}}{1-\phi_{i}^2}\Big) - 1\Bigg) &  i=k\\
         2\rho_{ik}^2 \exp \left[(\mu_{i} + \mu_{k}) + \dfrac{1}{2} \Big(\dfrac{\sigma_{i}^2}{1-\phi_{i}^2} + \dfrac{\sigma_{k}^2}{1-\phi_{k}^2}\Big)\right]& i \neq k
    \end{cases}
\end{equation}
\end{enumerate}
\end{proposition}

\begin{proof}

\end{proof}

\subsection{Bayesian Inference}
In this section, we describe Bayesian analysis for fitting the IR-MSV model in (\ref{Eq:Irregular_MSV}) to irregularly spaced multiple financial time series. We show the likelihood, prior, posterior and also discuss the computational aspects using MCMC algorithms. 
The likelihood function is 
\begin{equation}
    \mathcal{L}(\boldsymbol{\Theta}|\boldsymbol{r}_{t_{1}},\cdots,\boldsymbol{r}_{t_{T}}) = |\mathbf{\Sigma}_{\boldsymbol{r}}|^{-T/2}\prod_{j = 1}^{T}   \exp\Bigg(\dfrac{-\boldsymbol{r}_{t_{j}} \mathbf{\Sigma}_{\boldsymbol{r}}^{-1} \boldsymbol{r}_{t_{j}}^{\prime}}{2}\Bigg),
\end{equation}
where $\mathbf{\Sigma}_{\mathbf{r}} = \mathbf{H}_{t_{j}}^{1/2} \mathbf{R} \mathbf{H}_{t_{j}}^{1/2}$ ,  $\boldsymbol{\sigma} = (\sigma_{1},\cdots,\sigma_{p})$ and $\boldsymbol{\phi} = (\phi_{1},\cdots,\phi_{p})$, and 
$\boldsymbol{\Theta} = (\boldsymbol{\sigma}, \boldsymbol{\phi},\mathbf{R},\mu)$. 

For $i=1,\cdots,p$, we assume the following priors:
\begin{equation*}
\begin{aligned}
        &\mu_{i} \sim \text{N}(0,10), \\
        &\dfrac{1}{\sigma_{i}^2} \sim \text{G}(2.5,0.025), \\
        & \mathbf{R} \sim \text{LKJ}(\eta = 1.2), \text{ where }         \text{LKJ}(\mathbf{\Sigma}|\eta) \propto \det(\mathbf{\Sigma})^{(\eta - 1)}, \\ & \phi_{i} \sim \text{N}(0,0.5), 
\end{aligned}
\end{equation*}
where Lewandowski-Kurowicka-Joe (LKJ) distribution \cite{lewandowski2009generating} is a useful prior for correlation matrices, and IG denotes the inverse gamma distribution. 

The posterior distribution of $\boldsymbol{\Theta}$ is
\begin{equation}
    \begin{aligned}
        \pi(\boldsymbol{\Theta}| \boldsymbol{r}_{t_{1}},\cdots,\boldsymbol{r}_{t_{T}}) \propto & \quad |\mathbf{\Sigma}_{\boldsymbol{r}}|^{-T/2}\prod_{j = 1}^{T}   \exp\Bigg(\dfrac{-\boldsymbol{r}_{t_{j}} \mathbf{\Sigma}_{\boldsymbol{r}}^{-1} \boldsymbol{r}_{t_{j}}^{\prime}}{2}\Bigg) \\ &\times
        \prod_{j=1}^{T} \pi(\boldsymbol{h}_{t_{j}}|\boldsymbol{h}_{t_{j-1}},\boldsymbol{\sigma},\boldsymbol{\phi}) \times \pi(\boldsymbol{\sigma}) \pi(\boldsymbol{\phi}) \pi(\mathbf{R}) \pi(\mu), 
    \end{aligned}
\end{equation}
where we have assumed independence of the priors and $\pi(\boldsymbol{h}_{t_{j}}|\boldsymbol{h}_{t_{j-1}},\boldsymbol{\sigma},\boldsymbol{\phi})$ is the joint distribution of the stochastic volatilities as in  (\ref{Eq:Irregular_MSV1}). We employ a block random walk sampler for the correlation matrix $\mathbf{R}$ and for other parameters i.e. $(\boldsymbol{\sigma},\boldsymbol{\phi},\mu)$ we employ a Metropolis-Hastings adaptive random-walk sampler with a univariate normal proposal distribution, which are the default samplers in \texttt{NIMBLE} package \citep{de2017programming} in R.

\subsection{Simulation Study}
In this simulation our goal is to examine the effect of the correlation among the components on the accuracy of estimation. We generated 100 sets of zero mean log returns data from the IR-MSV model represented by equations (\ref{Eq:Irregular_MSV}) and (\ref{Eq:Irregular_MSV1}) with $p=3$ and $T=5000$. The gap times $g_{j}, 1 \leq j \leq T$ are generated using a similar procedure as described in Section~\ref{IR-SV_simulation}. The $3 \times 3$ correlation matrix has the following representation
\begin{equation*}
    \mathbf{R} = \begin{pmatrix}
        1 & \rho_{12} & \rho_{13}\\
        \rho_{12} & 1 & \rho_{23} \\
        \rho_{13} & \rho_{23} & 1
\end{pmatrix}
\end{equation*}
We describe three different scenarios. In scenario 1, we consider moderate positive correlations among the components with $\rho_{12} = 0.6$, $\rho_{13} = 0.4$ and $\rho_{23} = 0.2$. In scenario 2, we consider a mix of high and low positive correlation and a negative correlation with $\rho_{12} = - 0.4$, $\rho_{13} = 0.7$ and $\rho_{23} = 0.3$. In scenario 3, we consider all correlations to be equal, high and positive with $\rho_{12} = \rho_{13} = \rho_{23} = 0.7$.

We run 210000 MCMC iterations and discard the first 10000 as burn-in and thinned every $400^{th}$ sample to reduce autocorrelation between MCMC samples. We assessed the convergence for the parameters using trace and posterior density plots.
In Table
\ref{Table:Simulation_Scenario_1_IR_MSV}, Table \ref{Table:Simulation_Scenario_2_IR_MSV}, and Table \ref{Table:Simulation_Scenario_3_IR_MSV} we report the posterior sample mean along with their true values and $95\%$ credible intervals averaged the 100 data sets for scenario 1, scenario 2 and scenario 3 respectively.  

\begin{table}[H]
\begin{center}  
\begingroup
\setlength{\tabcolsep}{6pt} 
\renewcommand{\arraystretch}{0.6} 
\begin{tabular}{ccccc}
\hline \hline
Parameter & True Value & Mean & Q(2.5\%) & Q(97.5\%) \\
\hline \hline
$\rho_{12}$ & 0.6000 & 0.5987 & 0.5740 & 0.6158 \\
$\rho_{13}$ & 0.4000 & 0.3987 & 0.3724 & 0.4286 \\
$\rho_{23}$ & 0.2000 & 0.2008 & 0.1740 & 0.2254 \\
$\sigma^2_{1}$ & 1.0000 & 0.9831 & 0.8639 & 1.1203 \\
$\sigma^2_{1}$ & 0.8000 & 0.7807 & 0.6445 & 0.9119 \\
$\sigma^2_{3}$ & 0.5000 & 0.5658 & 0.4606 & 0.6636 \\
$\mu_{1}$ & -9.0000 & -8.9964 & -9.3681 & -8.6911 \\
$\mu_{2}$ & -9.5000 & -9.4927 & -9.6809 & -9.3255 \\
$\mu_{3}$ & -8.5000 & -8.5018 & -8.5973 & -8.3994 \\
$\phi_{1}$ & 0.7000 & 0.7021 & 0.6374 & 0.7694 \\
$\phi_{2}$ & 0.5000 & 0.4945 & 0.4119 & 0.5651 \\
$\phi_{3}$ & 0.3000 & 0.3009 & 0.2181 & 0.3943 \\
\hline \hline
\end{tabular}
\endgroup
\end{center}
\caption{True values and posterior estimates of parameters for Scenario 1 from the IR-MSV model.}
\label{Table:Simulation_Scenario_1_IR_MSV}
\end{table}

\begin{table}[H]
\begin{center}  
\begingroup
\setlength{\tabcolsep}{6pt} 
\renewcommand{\arraystretch}{0.6} 
\begin{tabular}{ccccc}
\hline \hline
Parameter & True Value & Mean & Q(2.5\%) & Q(97.5\%) \\
\hline \hline
$\rho_{12}$ & -0.4000 & -0.3993 & -0.4245 & -0.3749 \\
$\rho_{13}$ & 0.7000 & 0.7003 & 0.6835 & 0.7142 \\
$\rho_{23}$ & 0.3000 & 0.3003 & 0.2710 & 0.3268 \\
$\sigma^2_{1}$ & 1.0000 & 0.9778 & 0.8918 & 1.0933 \\
$\sigma^2_{2}$ & 0.8000 & 0.7867 & 0.6824 & 0.8885 \\
$\sigma^2_{3}$ & 0.5000 & 0.5613 & 0.4988 & 0.6305 \\
$\mu_{1}$ & -9.0000 & -9.0009 & -9.2870 & -8.6325 \\
$\mu_{2}$ & -9.5000 & -9.5044 & -9.6575 & -9.3313 \\
$\mu_{3}$ & -8.5000 & -8.4880 & -8.5960 & -8.4005 \\
$\phi_{1}$ & 0.7000 & 0.6928 & 0.6255 & 0.7494 \\
$\phi_{2}$ & 0.5000 & 0.4927 & 0.4120 & 0.5586 \\
$\phi_{3}$ & 0.3000 & 0.2954 & 0.2402 & 0.3541 \\ 
\hline \hline
\end{tabular}
\endgroup
\end{center}
\caption{True values and posterior estimates of parameters for Scenario 2 from the IR-MSV model.}
\label{Table:Simulation_Scenario_2_IR_MSV}
\end{table}

\begin{table}[H]
\begin{center}  
\begingroup
\setlength{\tabcolsep}{6pt} 
\renewcommand{\arraystretch}{0.6} 
\begin{tabular}{ccccc}
\hline \hline
Parameter & True Value & Mean & Q(2.5\%) & Q(97.5\%) \\
\hline \hline
$\rho_{12}$ & 0.7000 & 0.6985 & 0.6833 & 0.7147 \\
$\rho_{13}$ & 0.7000 & 0.6993 & 0.6859 & 0.7100 \\
$\rho_{23}$ & 0.7000 & 0.7006 & 0.6872 & 0.7139 \\
$\sigma^2_{1}$ & 1.0000 & 0.9865 & 0.9045 & 1.0815 \\
$\sigma^2_{2}$ & 0.8000 & 0.7811 & 0.7305 & 0.8642 \\
$\sigma^2_{3}$ & 0.5000 & 0.4885 & 0.4241 & 0.5620 \\
$\mu_{1}$ & -9.0000 & -8.9865 & -9.1761 & -8.7485 \\
$\mu_{2}$ & -9.5000 & -9.4942 & -9.6221 & -9.3748 \\
$\mu_{3}$ & -8.5000 & -8.4971 & -8.6012 & -8.3849 \\
$\phi_{1}$ & 0.7000 & 0.6998 & 0.6351 & 0.7572 \\
$\phi_{2}$ & 0.5000 & 0.4912 & 0.3955 & 0.5725 \\
$\phi_{3}$ & 0.5000 & 0.4965 & 0.3965 & 0.5785 \\
\hline \hline
\end{tabular}
\endgroup
\end{center}
\caption{True values and posterior estimates of parameters for Scenario 3 from the IR-MSV model.}
\label{Table:Simulation_Scenario_3_IR_MSV}
\end{table}

\section{Data Analysis: Multiple Intraday Log Returns}
We 
implement the proposed models 
on intraday log returns for three health stocks traded on the NYSE: Medtronic PLC (MDT), Bristol-Myers Squibb Co (BMY), and CVS Health Corp (CVS). In particular, we compare the 
forecast performance of the IR-MSV model with the IR-SV models fit to each stock.

\subsection{Data Description}
We considered the 
intra-day prices for MDT, BMY and CVS traded on June 24, 2016. We aggregated the high frequency prices at the one second level by taking the price with respect to the latest time point as in \cite{buccheri2021high}. 
There were 10066, 8797 and 11330 observations (transactions) for MDT, BMY, and CVS respectively. We synchronized these 
prices from three stocks using the refresh time sampling approach \citep{barndorff2011multivariate} to obtain a sample size of $T=6239$ synchronized observations. A brief description of refresh sampling technique is discussed in (Section 2.5.3). Using these irregularly spaced refreshed prices ($P_{t_{j}}$) we calculate the log returns as
\begin{equation*}
    r_{t_{j}} = \log(P_{t_{j}}) - \log(P_{t_{j-1}}), \quad 2\leq j \leq T
\end{equation*}
We have plotted the time series of log returns calculated from the refreshed prices of MDT, BMY and CVS on $24^{th}$ June, 2016 in Figure \ref{fig:log returns}. They seem to be fairly correlated.  
\begin{figure}[H]
    \centering
    \includegraphics[scale=0.70]{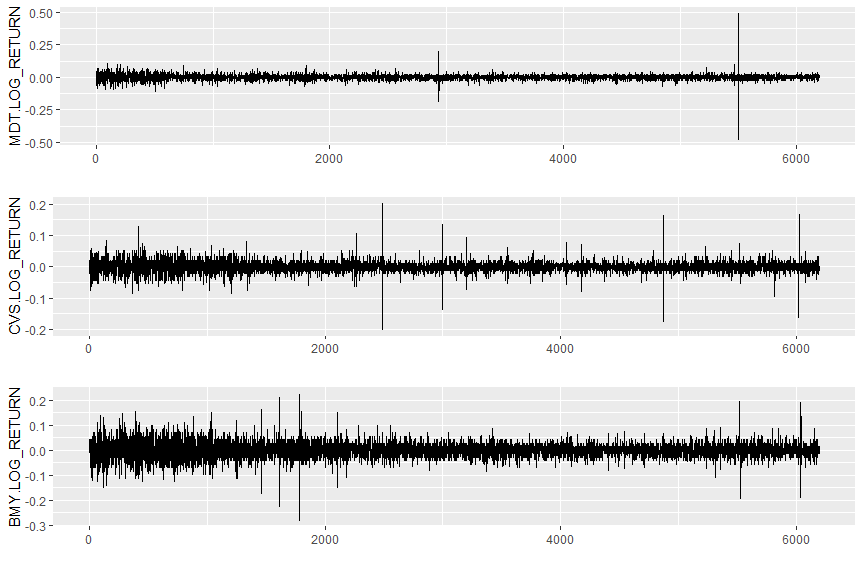}
    \caption{Log returns of BMY, CVS and MDT calculated from the refreshed prices on $24^{th}$ June, 2016}
    \label{fig:log returns}
\end{figure}

\subsection{Results}
We fit IR-MSV model as in equations (\ref{Eq:Irregular_MSV}) and (\ref{Eq:Irregular_MSV1}) to the irregularly spaced multivariate intraday log returns series of three stocks. We use first 6195 observations for fitting IR-MSV model and kept 44 observations as hold out samples. We have used exactly the same priors for the parameters as in the simulation setup. We run 300,000 MCMC iterations and discard the first 50,000 as burn-in to obtain samples from every $500^{th}$ iteration of the last 250,000 iterations. For convergence of MCMC iterations we have used trace and posterior density plots. In Table \ref{Table:IR_MSV_real_data}, we report the posterior sample mean along with their $95\%$ credible intervals. 
\begin{table}[H]
\begin{center}  
\begingroup
\setlength{\tabcolsep}{6pt} 
\renewcommand{\arraystretch}{0.6} 
\begin{tabular}{ccccc}
\hline \hline
Parameter & Mean & SD & 2.50\% & 97.50\% \\
\hline \hline
$\rho_{12}$ & 0.276 & 0.013 & 0.249 & 0.301 \\
$\rho_{13}$ & 0.352 & 0.012 & 0.326 & 0.372 \\
$\rho_{23}$ & 0.292 & 0.013 & 0.267 & 0.317 \\
$\sigma^2_{1}$ & 0.388 & 0.046 & 0.302 & 0.480 \\
$\sigma^2_{2}$ & 0.715 & 0.057 & 0.610 & 0.837 \\
$\sigma^2_{3}$ & 0.442 & 0.048 & 0.352 & 0.542 \\
$\mu_{1}$ & -7.286 & 0.054 & -7.390 & -7.176 \\
$\mu_{2}$ & -8.589 & 0.051 & -8.690 & -8.481 \\
$\mu_{3}$ & -8.286 & 0.052 & -8.395 & -8.190 \\
$\phi_{1}$ & 0.403 & 0.062 & 0.286 & 0.529 \\
$\phi_{2}$ & 0.143 & 0.042 & 0.077 & 0.236 \\
$\phi_{3}$ & 0.299 & 0.059 & 0.190 & 0.409 \\
\hline \hline 
\end{tabular}
\endgroup
\end{center}
\caption{Table showing posterior estimates and $95\%$ credible intervals from IR-MSV model for the intraday log returns of the three stocks.}
\label{Table:IR_MSV_real_data}
\end{table}
All parameters lie within the $95\%$ credible intervals limits. The estimated correlations are moderate indicating fair dependency among the components which is expected. The persistence parameter for BMY is much higher than that of CVS and MDT indicating the temporal correlation of volatility of BMY is much higher than the temporal correlation of volatility of CVS and MDT.  

We have also fitted IR-SV models to the univariate refreshed log return series of each stock of length $T=6195$. We have used the same prior as discussed in the simulation study of IR-SV models in Section \ref{IR-SV_simulation}. We run 550,000 MCMC iterations and discard the first 50,000 as burn-in to obtain samples from every 1000th iteration of the last 500,000 iterations. In Table \ref{Table:IR_SV_MDT}, Table \ref{Table:IR_SV_BMY} and Table \ref{Table:IR_SV_CVS} we report the posterior sample mean along with their $95\%$ credible intervals. All the parameters lie inside the $95\%$ credible intervals. Convergence for the parameters are assessed using trace and posterior density plots. We observe that the persistence parameter $\phi$ for BMY is higher than that of CVS and MDT which is consistent with the result obtained in case of IR-MSV model fitting.
\begin{table}[H]
\begin{center}  
\begingroup
\setlength{\tabcolsep}{6pt} 
\renewcommand{\arraystretch}{0.6} 
\begin{tabular}{ccccc}
\hline \hline
Parameter & Mean & SD & 2.50\% & 97.50\% \\
\hline \hline
$\mu$ & -7.248 & 0.084 & -7.413 & -7.088 \\
$\phi$ & 0.670 & 0.065 & 0.539 & 0.793 \\
$\sigma$ & 0.486 & 0.042 & 0.403 & 0.568 \\ 
\hline \hline
\end{tabular}
\endgroup
\end{center}
\caption{Table showing posterior estimates and $95\%$ credible intervals from IR-SV model for the refreshed intraday log returns of BMY.}
\label{Table:IR_SV_BMY}
\end{table}

\begin{table}[H]
\begin{center}  
\begingroup
\setlength{\tabcolsep}{6pt} 
\renewcommand{\arraystretch}{0.6} 
\begin{tabular}{ccccc}
\hline \hline
Parameter & Mean & SD & 2.50\% & 97.50\% \\
\hline \hline
$\mu$ & -8.571 & 0.053 & -8.672 & -8.468 \\
$\phi$ & 0.170 & 0.048 & 0.085 & 0.271 \\
$\sigma$ & 0.829 & 0.034 & 0.763 & 0.898 \\
\hline \hline
\end{tabular}
\endgroup
\end{center}
\caption{Table showing posterior estimates and $95\%$ credible intervals from IR-SV model for the refreshed intraday log returns of CVS.}
\label{Table:IR_SV_CVS}
\end{table}

\begin{table}[H]
\begin{center}  
\begingroup
\setlength{\tabcolsep}{6pt} 
\renewcommand{\arraystretch}{0.6} 
\begin{tabular}{ccccc}
\hline \hline
Parameter & Mean & SD & 2.50\% & 97.50\% \\
\hline \hline
$\mu$ & -8.267 & 0.059 & -8.381 & -8.148 \\
$\phi$ & 0.381 & 0.077 & 0.231 & 0.528 \\
$\sigma$ & 0.648 & 0.039 & 0.570 & 0.724 \\
\hline \hline
\end{tabular}
\endgroup
\end{center}
\caption{Table showing posterior estimates and $95\%$ credible intervals from IR-SV model for the refreshed intraday log returns of MDT.}
\label{Table:IR_SV_MDT}
\end{table}

\section{Summary and Discussion}
In this chapter we extend gap time modeling idea of \cite{nagaraja2011autoregressive} to propose univariate volatility models for irregular financial time series by modifying the regularly spaced stochastic volatility models. We also extend this approach to propose multivariate stochastic volatility (MSV) models for multiple
irregularly spaced time series which is a modification of the MSV model of \cite{chib2009multivariate} that was used with daily data. We used simulation studies to show the accuracy of estimation of the proposed models. We have used the proposed models to model intraday logarithmic returns for three health sector NYSE stocks and compared the univariate and multivariate forecasting performance on shorter and longer horizons.

We can construct an IR-GARCH model in a similar way. 
We have illustrated the accuracy of estimation in IR-GARCH(1,1) models using simulations. This idea can also be extended to multivariate GARCH models which has been left for future research. 

\section{Appendix}

\subsection{Explorations with Irregular GARCH models}

One of the early approaches to model the return series sampled at irregularly spaced time intervals set by the trade arrivals is ACD-GARCH model by \cite{ghysels1998garch}. This is a random coefficient GARCH, or doubly stochastic GARCH, where the stochastic durations between transactions determine the parameter dynamics for the GARCH equation. They proposed a Generalized Method of Moments (GMM) based two-step procedure for estimating the parameters. ACD-GARCH is quite cumbersome to estimate and also difficult to generalize to multiple time series. \cite{meddahi2006garch} proposed a GARCH type model for irregularly spaced data which is an exact discretization of continuous time stochastic volatility processes observed at irregularly spaced times. Their model combines the advantages of ACD-GARCH \cite{ghysels1998garch} and ACD \citep{Engle1998}. A continuous version of GARCH (COGARCH) is another way to model irregularly spaced time series data \citep{maller2008garch}. Recently, \cite{buccheri2021score} propose to model intraday log-prices through a multivariate local-level model with score-driven covariance matrices and to treat asynchronicity as a missing value problem.

In this section we do some preliminary investigation of extending  the gap time autoregressive model of \cite{nagaraja2011autoregressive} to develop an irregular GARCH (IR-GARCH) model. 
We introduce the IR-GARCH model and discuss its properties. 
We conduct a small simulation study to assess the accuracy of estimation under different scenarios. Future research of this topic can be useful.

Let $\{r_{t_{j}}\}$ be a sequence of zero mean log-returns of an asset. Here $t_{j}$ denotes the time of the $j^{th}$ observation and $g_j = t_{j} - t_{j-1}$, $j>1$ is the known $j^{th}$ observed gap time. Suppose $\{\epsilon_{t_{j}}\}$ is a sequence of real valued independent and identically distributed
random variables, having mean 0 and variance 1. 
We start with an IR-GARCH(1,1) model as shown below.

\subsubsection*{IR-GARCH(1,1)}

Then, $r_{t_{j}}$ follows IR-GARCH$(1,1)$ model if
\begin{equation}
\label{Eq:IR-GARCH}
    r_{t_{j}} = \sigma_{t_{j}} \epsilon_{t_{j}}
\end{equation}
\begin{equation}
\begin{gathered}
\label{Eq:latent_IR-GARCH}
    \sigma_{t_{1}}^2 = \omega(1-\alpha_1^{g_{1}} -\beta_1^{g_1}),\\
    \sigma_{t_{j}}^2 = \omega(1-\alpha_1^{g_{j}} -\beta_1^{g_{j}}) + \alpha_1^{g_{j}}r_{t_{j-1}}^2 + \beta_1^{g_j}\sigma_{t_{j-1}}^2, \text{ for $j>1$.}
\end{gathered}
\end{equation}
Let $g* = \min_{j}g_j$; then the constraints for $\sigma_{t_{j}}^2$ to be positive are $\omega >0, \alpha_1 >0, \beta_1 >0$ and $\alpha_1^{g*}+\beta_1^{g*} <1$.
The unconditional mean of $r_{t_{j}}$ is
\begin{align*}  
    \text{E}(r_{t_{j}}) = \text{E}\big(\text{E}\big(r_{t_{j}}|\mathcal{F}_{t_{j-1}}\big)\big)= \text{E}\big(\text{E}\big(\sigma_{t_{j}} \epsilon_{t_{j}}|\mathcal{F}_{t_{j-1}}\big)\big)= \text{E}\big(\sigma_{t_{j}} \text{E}\big(\epsilon_{t_{j}}|\mathcal{F}_{t_{j-1}}\big)\big)=0, \text{$\forall j$.}
\end{align*}
The unconditional variance of $r_{t_{j}}$ is
\begin{equation} \label{Eq:variance_IR_GARCH_t1}
    \text{Var}(r_{t_{1}}) = \text{E}(r^2_{t_{1}}) =  \omega(1-\alpha_1^{g_1} -\beta_1^{g_1}), \text{ since $E(\epsilon_{t_{1}}^2)=1$}
\end{equation}
and for $j>1$,
\begin{align*}
    \text{Var}(r_{t_{j}}) = \text{E}(r^2_{t_{j}}) = \text{E}[\text{E}[r^2_{t_{j}}| \mathcal{F}_{t_{j-1}}]] = E[\sigma^2_{t_{j}}] = \omega(1-\alpha_1^{g_j} -\beta_1^{g_j}) + \alpha_1^{g_j} \text{E}[r_{t_{j-1}}^2] + \beta_1^{g_j} \text{E}[\sigma_{t_{j-1}}^2]
\end{align*}
Assuming stationarity of $\{r_{t_{j}}\}$ with $\text{E}(r_{t_{j}}) =0$, $\text{Var}(r_{t_{j}}) = \text{Var}(r_{t_{j-1}})=\text{E}(r_{t_{j-1}}^2)$
we have 
\begin{equation}\label{Eq:variance_IR_GARCH}
    \text{Var}(r_{t_{j}}) = \omega, \text{ provided $\alpha_1^{g*}+\beta_1^{g*} <1$.
 }
\end{equation}

\subsubsection*{IR-ARCH(1) model}

The returns $r_{t_{j}}$ are said to follow an IR-ARCH(1) model if
\begin{equation}
\begin{gathered}
\label{Eq:IR-ARCH}
    r_{t_{j}} = \sigma_{t_{j}} \epsilon_{t_{j}},  \\
    \sigma_{t_{1}}^2 = \omega(1-\alpha_1^{g_1}),\\
    \sigma_{t_{j}}^2 = \omega(1-\alpha_1^{g_j}) + \alpha_1^{g_j}r_{t_{j-1}}^2 \text{, for $j>1$}, 
\end{gathered}
\end{equation}
where $g_{j}$'s are the observed gap time between two consecutive observations and $\{\epsilon_{t_{j}}\}$ is a sequence of real valued independent and identically distributed random variables with mean 0 and variance 1.
\begin{proposition}\label{Eq:IRARCH_theorem}
Suppose $\{r_{t_{j}}\}$ follows an IR-ARCH(1) model defined by (\ref{Eq:IR-ARCH}). Let $\eta_{t_{1}} = r_{t_{1}}^2 - \sigma_{t_{1}}^2 - \omega \alpha_{1}^{g_{1}}$ and  $\eta_{t_{j}} = r_{t_{j}}^2 - \sigma_{t_{j}}^2$, for $j>1$. Suppose $E[\eta_{t_{1}}^2] = C$ and for $j>1$, $E[\eta_{t_{j}}^2] = C(1-\alpha_1^{2g_j})$, where $C>0$ is a constant. Then $r_{t_{j}}^2$ is a stationary process.
\end{proposition}

\begin{proof}
We have 
\begin{equation*}
    \begin{aligned}
            \eta_{t_{1}} = r_{t_{1}}^2 - \sigma_{t_{1}}^2 - \omega \alpha_{1}^{g_{1}}\\
    \Rightarrow \sigma_{t_{1}}^2 = r_{t_{1}}^2 - \omega \alpha_{1}^{g_{1}} - \eta_{t_{1}}
    \end{aligned}
\end{equation*}
Substituting in equation (\ref{Eq:IR-ARCH}) we get,
\begin{equation*}
    r_{t_{1}}^2 - \omega = \eta_{t_{1}},
\end{equation*}
where $E(\eta_{t_{1}}) = E(r_{t_{1}}^2) - E(\sigma_{t_{1}}^2) = 0$, it follows from equation(\ref{Eq:variance_IR_GARCH_t1}).
Similarly, for $j > 1$,
\begin{equation*}
    \begin{aligned}
        &r_{t_{j}}^2 - \eta_{t_{j}} = \omega (1-\alpha_{1}^{g_{j}}) + \alpha_{1}^{g_{j}} r_{t_{j-1}}^2 \\
        &\Rightarrow r_{t_{j}}^2 - \omega = \alpha_{1}^{g_{j}}(r_{t_{j-1}}^2 - \omega) + \eta_{t_{j}}.
    \end{aligned}
\end{equation*}
where
$E(\eta_{t_{j}}) = E(r_{t_{j}}^2 - \sigma_{t_{j}}^2) = E(E(\sigma_{t_{j}}^2 \epsilon_{t_{j}}^2|\mathcal{F}_{t_{j-1}})) - E(\sigma_{t_{j}}^2) = E(\sigma_{t_{j}}^2) - E(\sigma_{t_{j}}^2) = 0$, for $j>1$.
Let $x_{t_{j}} = r_{t_{j}}^2 - \omega$, along with $E[\eta_{t_{1}}^2] = C$ and for $j>1$, $E[\eta_{t_{j}}^2] = C(1-\alpha_1^{2g_j})$, where $C>0$ is a constant.
Using Propostion 5.1 
\cite{nagaraja2011autoregressive}, 
it follows that $x_{t_{j}}$ and hence $r_{t_{j}}^2$ is a stationary process.

\end{proof}

\subsubsection*{Estimation}
We will use conditional maximum likelihood estimation \citep{tsay2005analysis} by assuming standard normal distribution for errors for estimating the parameters. Let $\boldsymbol{\theta} = (\omega,\alpha_1,\beta_1)'$ be the vector of parameters for instance, assuming $\epsilon_{t_{j}}$ to follow a standard normal distribution, the conditional Gaussian log-likelihood is given by 
\begin{equation}
    l_n(\boldsymbol{\theta}) = - \frac{1}{2} \sum_{j=2}^{n} \Bigg[\log(2\pi) + \log(\sigma_{t_{j}}^2) + \dfrac{r_{t_{j}}^2}{\sigma_{t_j}^2}\Bigg]
    ,
\end{equation}
where $\sigma_{t_j}^2$ is updated according to (\ref{Eq:latent_IR-GARCH}) for the IR-GARCH model. 
The estimated parameter vector is 
\begin{equation}
    \begin{aligned}
        \hat{\boldsymbol{\theta}} = \max_{\boldsymbol{\theta}} l_n(\boldsymbol{\theta}).
    \end{aligned}
\end{equation}

\subsection{IR-GARCH(1,1)} \label{IR-GARCH simulation}
We describe two scenarios. For each scenario, we generated 100 sets of zero mean log returns data, each of length $T = 5000$. The gap times $g_{j}, 1 \leq j \leq T$ are generated using a similar procedure as described in Section~\ref{IR-SV_simulation}. The following Table \ref{Table:Simulation_Scenario_1_IR_GARCH} and Table \ref{Table:Simulation_Scenario_2_IR_GARCH} give the true values and conditional ML estimates of parameters for Scenario 1 and Scenario 2.
\begin{table}[H]
\begin{center}  
\begin{tabular}{ccccc}
\hline \hline
Parameter & True Value & Estimate & 2.5\% & 97.5\% \\
\hline \hline
$\omega$ & 0.0100 & 0.0102 & 0.0088 & 0.0115 \\
$\alpha$ & 0.7000 & 0.7070 & 0.6479 & 0.7662 \\
$\beta$ & 0.2500 & 0.2427 & 0.1511 & 0.3343\\
\hline \hline
\end{tabular}
\end{center}
\caption{True values and conditional ML estimates of parameters for Scenario 1 from the IR-GARCH model.}
\label{Table:Simulation_Scenario_1_IR_GARCH}
\end{table}
\begin{table}[H]
\begin{center}  
\begin{tabular}{ccccc}
\hline \hline
Parameter & True Value & Estimate & 2.5\% & 97.5\% \\
\hline \hline
$\omega$ & 0.0100 & 0.0106 & 0.0073 & 0.0138 \\
$\alpha$ & 0.9000 & 0.9053 & 0.8682 & 0.9424 \\
$\beta$ & 0.0500 & 0.0475 & 0.0130 & 0.0820\\
\hline \hline
\end{tabular}
\end{center}
\caption{True values and conditional ML estimates of parameters for Scenario 2 from the IR-GARCH model.}
\label{Table:Simulation_Scenario_2_IR_GARCH}
\end{table}
We have used \texttt{R} package \texttt{Rsolnp} to optimize the conditional maximum likelihood function to obtain the parameters in IR-GARCH(1,1) model. All the true values of the parameters are inside the $95\%$ confidence interval.